\documentclass{article}

\usepackage{authblk}
\usepackage{graphicx}
\usepackage{multirow}
\usepackage{amsmath,amssymb,amsfonts}
\usepackage{amsthm}
\usepackage{mathrsfs}
\usepackage[margin=1in]{geometry}
\usepackage[title]{appendix}
\usepackage{xcolor}
\usepackage{textcomp}
\usepackage{manyfoot}
\usepackage{booktabs}
\usepackage{algorithm}
\usepackage{algorithmicx}
\usepackage{algpseudocode}
\usepackage{listings}
\usepackage{thm-restate}
\usepackage{hyperref}
\usepackage{cases}
\usepackage{microtype}
\usepackage[T1]{fontenc}
\usepackage{diagbox}
\usepackage{url}
\usepackage{cleveref}
\usepackage{lmodern}
\usepackage{tabularx}

\theoremstyle{plain}\newtheorem{theorem}{Theorem}
\newtheorem{corollary}{Corollary}
\newtheorem{lemma}{Lemma}

\newtheorem{example}{Example}

\theoremstyle{definition}\newtheorem{definition}{Definition}

\newcommand{\MaxPalE}{\mathit{MaxPalEnd}}
\newcommand{\lca}{\mathsf{lca}}

\newcommand{\LPS}{\mathsf{LPS}}

\newcommand{\rev}[1]{{#1}^{\mathit{R}}}
\newcommand{\parent}{\mathsf{parent}}
\newcommand{\str}{\mathsf{str}}
\newcommand{\EERTREE}{\mathsf{EERTREE}}

\newcommand{\slink}{\mathsf{slink}}
\newcommand{\qlink}{\mathsf{qlink}}
\newcommand{\dlink}{\mathsf{dlink}}

\newcommand{\Color}{\mathsf{color}}
\newcommand{\Uncolor}{\mathsf{uncolor}}
\newcommand{\Insert}{\mathsf{insert}}
\newcommand{\Delete}{\mathsf{delete}}
\newcommand{\Pred}{\mathsf{pred}}
\newcommand{\Succ}{\mathsf{succ}}

\newcommand{\InsertLeaf}{\mathsf{insert\_leaf}}
\newcommand{\DeleteLeaf}{\mathsf{delete\_leaf}}
\newcommand{\NCA}{\mathsf{NCA}}

\newcommand{\Leaves}{{L}}

\newcommand{\timeP}{\mathop{\mathsf{P}_{\mathsf{all}}}}
\newcommand{\timePincr}{\mathop{\mathsf{P}_{\mathsf{incr}}}}
\newcommand{\timeCP}{\mathop{\mathsf{CP}_{\mathsf{all}}}}
\newcommand{\timeCPincr}{\mathop{\mathsf{CP}_{\mathsf{incr}}}}

\begin{document}

\title{Online Computation of Palindromes and Suffix Trees on Tries}

\author[1]{Hiroki Shibata\thanks{\texttt{shibata.hiroki.753@s.kyushu-u.ac.jp}}}
\author[2]{Mitsuru Funakoshi\thanks{\texttt{mitsuru.funakoshi0000@gmail.com}}}
\author[3]{Takuya Mieno\thanks{\texttt{tmieno@uec.ac.jp}}}
\author[2]{Masakazu Ishihata\thanks{\texttt{masakazu.ishihata@ntt.com}}}
\author[4]{Yuto Nakashima\thanks{\texttt{nakashima.yuto.003@m.kyushu-u.ac.jp}}}
\author[4]{Shunsuke Inenaga\thanks{\texttt{inenaga.shunsuke.380@m.kyushu-u.ac.jp}}}
\author[5]{Hideo Bannai\thanks{\texttt{hdbn.dsc@tmd.ac.jp}}}
\author[4]{Masayuki Takeda}

\affil[1]{Joint Graduate School of Mathematics for Innovation, Kyushu University, Japan}
\affil[2]{NTT Communication Science Laboratories, Kyoto, Japan}
\affil[3]{Department of Computer and Network Engineering, University of Electro-Communications, Japan}
\affil[4]{Department of Informatics, Kyushu University, Japan}
\affil[5]{M\&D Data Science Center, Institute of Integrated Research, Institute of Science Tokyo, Japan}

\maketitle

\abstract{
  We consider the problems of computing maximal palindromes and distinct palindromes
  in a trie.
  A trie is a natural generalization of a string, which can be seen as a single-path tree.
  There is a linear-time offline algorithm to compute
  maximal palindromes and distinct palindromes in a given (static) trie
  whose edge-labels are drawn from a linearly-sortable alphabet [Mieno et al., ISAAC 2022].
  In this paper, we tackle problems of palindrome enumeration on dynamic tries
  which support leaf additions and leaf deletions.
  We propose the first sub-quadratic algorithms to enumerate palindromes in a dynamic trie.
  For maximal palindromes, we propose an algorithm that runs in $O(N \min(\log h, \sigma))$ time
  and uses $O(N)$ space, where
  $N$ is the maximum number of edges in the trie,
  $\sigma$ is the size of the alphabet,
  and $h$ is the height of the trie.
  For distinct palindromes, we develop several online algorithms based on different algorithmic frameworks,
  including approaches using the EERTREE (a.k.a. palindromic tree) and the suffix tree of a trie.
  These algorithms support leaf insertions and deletions in the trie
  and achieve different time and space trade-offs.
  Furthermore, as a by-product,
  we present online algorithms to construct the suffix tree and the EERTREE of the input trie,
  which is of independent interest.
}

\section{Introduction}\label{sec:intro}

\emph{Palindromes} are strings that read the same forward and backward.
Finding palindromic structures in a given string is a fundamental task in string processing and has been studied extensively (e.g., see~\cite{Manacher75, Apostolico1995parallel, MatsubaraIISNH09, DBLP:journals/ipl/GroultPR10, DBLP:conf/stringology/KosolobovRS13, Porto2002ApproxPalindrome, KolpakovK09, NarisadaDNIS20, GawrychowskiIIK18} and references therein).

Consider the set $C_n = \{1,1.5,2, \ldots, n\}$ of
$2n-1$ half-integer and integer positions in a string $T$ of length $n$.
The \emph{maximal palindrome} for a position $c \in C_n$ in $T$
is a non-extensible palindrome whose center lies at $c$.
It is easy to store all maximal palindromes with $O(n)$ total space;
e.g., simply store their lengths in an array of length $2n-1$
together with the input string $T$.
If $P = T[i..j]$ is a maximal palindrome with center $c = \frac{i+j}{2}$,
then clearly all substrings $P' = T[i+d..j-d]$ with $0 \leq d \leq \frac{j-i}{2}$
are also palindromes.
Hence, by computing and storing all maximal palindromes in $T$,
we can obtain a compact representation of all palindromes in $T$.
Manacher~\cite{Manacher75} gave
an elegant $O(n)$-time algorithm to compute all maximal palindromes in $T$.
This algorithm works over a general unordered alphabet.
For the case where the input string is drawn from a
linearly sortable alphabet~\cite{EllertGG23},
such as an integer alphabet of size polynomial in $n$,
there exists an alternative suffix-tree-based algorithm~\cite{Weiner73}
that runs in $O(n)$ time~\cite{gusfield97:_algor_strin_trees_sequen}.
The length of the maximal palindrome centered at position $c$
can be determined by computing the longest common prefix (LCP)
of $T[\lfloor c \rfloor .. n]$ and $\rev{T[\lceil c \rceil .. 1]}$.
In this approach, a suffix tree of $T \# \rev{T} \$$ is constructed,
where $\rev{T}$ denotes the reversed string of $T$,
and $\#$ and $\$$ are special symbols that do not appear in $T$.
By augmenting the suffix tree with a lowest common ancestor (LCA)
data structure~\cite{Farach-ColtonFM00},
the LCP length can be computed in $O(1)$ time
after an $O(n)$-time preprocessing step.
Thus, for each center, the maximal palindrome can be obtained in constant time,
and the computation over all centers can be done in linear time.

Another central notion regarding palindromic substrings is
\emph{distinct palindromes}.
Droubay et al.~\cite{DBLP:journals/tcs/DroubayJP01} showed that any string of length $n$
contains at most $n+1$ distinct palindromes
(including the empty string).
Groult et al.~\cite{DBLP:journals/ipl/GroultPR10}
proposed an $O(n)$-time algorithm for computing all distinct palindromes in a string
of length $n$ over a linearly-sortable alphabet.
Moreover, all distinct palindromes of a string can be represented compactly by the
\emph{EERTREE} (a.k.a.~palindromic tree)~\cite{EERTREE}.

In this paper, instead of strings, we consider
a \emph{trie} $\mathcal{T}$ as input,
where each edge is labeled by a single character from the alphabet $\Sigma$
and the outgoing edges of each node are labeled by mutually distinct characters.
A trie is a natural extension of a string,
and is a compact representation of a set of strings.
There are a number of works on efficient algorithms on tries,
such as indexing a (backward) trie~\cite{Breslauer1998,Kosaraju89a,Shibuya03,inenaga01:_const_cdawg_trie,MohriMW09,ferragina09:_compr,NakashimaIIBT15,Inenaga20} for exact pattern matching,
parameterized pattern matching on a trie~\cite{AmirN09,FujisatoNIBT19},
order preserving pattern matching on a trie~\cite{NakamuraIBT17},
and finding all maximal repetitions (a.k.a. runs) in a trie~\cite{SugaharaNIBT21}.
Recently,
a linear-time algorithm to compute maximal and distinct palindromes in a trie
was proposed in~\cite{MienoFI22}.
We emphasize that the algorithm of~\cite{MienoFI22} works \emph{offline},
and that its time complexity depends on the assumption of a linearly-sortable alphabet,
while our algorithm described below works in an \emph{online} manner
and over a general ordered alphabet.

This paper tackles the problems of computing
all maximal palindromes and all distinct palindromes in a trie $\mathcal{T}$
given in an \emph{online} manner,
i.e., the trie grows step-by-step by inserting one leaf per one step.
Na\"ive methods for solving these problems would be to apply
Manacher's algorithm~\cite{Manacher75} or
Groult et al.'s algorithm~\cite{DBLP:journals/ipl/GroultPR10}
for each root-to-leaf path string in $\mathcal{T}$,
but this requires $\Omega(N^2)$ time in the worst case
since there exists a trie with $N$ edges that
can represent $\Theta(N)$ strings of length $\Theta(N)$ each.
We also remark that a direct application of Manacher's algorithm to a trie
does not seem to solve our problem efficiently,
since the amortization argument in the case of a single string
does not hold in our case of a trie.
The aforementioned suffix tree approach of Gusfield~\cite{gusfield97:_algor_strin_trees_sequen}
cannot be applied to our trie case either,
even for offline algorithm design:
While the number of suffixes in the reversed leaf-to-root direction
of the trie $\mathcal{T}$ is $N$,
the number of suffixes in the forward root-to-leaf direction can be $\Theta(N^2)$
in the worst case.
Thus, one cannot afford to construct a suffix tree
that contains all suffixes of the forward paths of $\mathcal{T}$.

In this paper, we first show that
the number of maximal palindromes in a trie $\mathcal{T}$ with $N$ edges and $L$ leaves
is exactly $2N-L$ and that
the number of distinct palindromes in $\mathcal{T}$ is at most $N+1$.
These results generalize the known bounds for a single string.
Then, we present an online algorithm to compute all maximal palindromes
that runs in $O(N \min(\log h, \sigma))$ time and $O(N)$ space
by exploiting the periodicity of palindromes,
where $h$ is the height of the input trie $\mathcal{T}$.
Our algorithm should be compared with existing results of~\cite{MienoFI22} for offline setting.
Their algorithm first constructs the suffix tree of a backward trie~\cite{Kosaraju89a},
and thus, 
the total time complexity increases to 
$O(N\log\sigma)$ time for a general ordered alphabet or
$O(N\sigma)$ time for a general unordered alphabet
due to the sorting complexity~\cite{Farach-ColtonFM00}.
Since $h$ is incomparable with $\sigma$ in general,
$O(N \min(\log h, \sigma))$ is incomparable with $O(N\log\sigma)$.
On the other hand, for a general unordered alphabet,
our algorithm never performs worse than $O(N\sigma)$.

Furthermore, we present several online algorithms for computing all distinct palindromes in a trie.
These algorithms fall into two main frameworks.
The first framework is based on the online construction of the suffix tree of a backward trie,
while the second framework is based on the online construction of the EERTREE of a forward trie.
Within the EERTREE-based framework, we describe multiple variants that differ
in the way palindromic suffixes are traversed and maintained.
All proposed algorithms operate over a general ordered alphabet.
However, some of the algorithms admit improved time bounds
under specific condition on the alphabet.
The online construction techniques developed in this work,
including both the suffix-tree-based and the EERTREE-based approaches,
may be of independent interest.
The algorithms and their time and space complexities are summarized in Table~\ref{tab:online_algorithms}.

\begin{table}[tb]
  \centering
  \caption{
    A summary of the online algorithms for computing distinct palindromes in a trie
    presented in this paper.
    We denote 
    by $N$ the number of edges in the trie, 
    by $D$ the number of distinct palindromes in the trie,
    by $h$ its height, and 
    by $\sigma$ the alphabet size.
    The term $\timePincr(\sigma) \in O(\log \sigma)$ denotes the query time of a linear-space incremental predecessor data structure
    over a universe of size $\sigma$.
    The term $\timeCPincr(N, \sigma) \in O(\log N)$ denotes the query time of a linear-space incremental colored predecessor data structure
    with $\sigma$ colors and $N$ total elements, as described in Lemma~\ref{lem:colored_pred}.
    In general, $\timeCPincr(N, \sigma)$ is not directly comparable to $\log h$ or $\log \sigma$.
    By employing a colored predecessor data structure that supports deletions,
    each algorithm can be extended to handle leaf deletions.
  }
  \vspace{\baselineskip}
  \label{tab:online_algorithms}
  \begin{tabular}{@{}llcc@{}}
    \toprule
    \textbf{Framework} & \textbf{Method} & \textbf{Time Complexity} & \textbf{Space Complexity} \\
    \midrule
    \multirow{8}{*}{\shortstack[l]{EERTREE\\(Forward Trie)}}
& \begin{tabular}[c]{@{}l@{}}
          Quick Link \\
          (Cor.~\ref{cor:eertree_online_qlink})
        \end{tabular}
      & $O(N \min(\log h, \sigma) + D \timePincr(\sigma))$
      & $O(N)$ \\
      \cmidrule(l){2-4}
      
& \begin{tabular}[c]{@{}l@{}}
          Direct Link with \\
          Persistent Tree \\
          (Cor.~\ref{cor:eertree_online_dlink_persistent})
        \end{tabular}
      & $O(N \log \min(\log h, \sigma) + D \timePincr(\sigma))$
      & $O(N + D \log \min(\log h, \sigma))$ \\
      \cmidrule(l){2-4}
      
& \begin{tabular}[c]{@{}l@{}}
          Direct Link with \\
          Colored Ancestor \\
          (Cor.~\ref{cor:eertree_online_dlink_nca})
        \end{tabular}
      & $O(N (\timeCPincr(D, \sigma)))$ 
      & $O(N)$ \\
    \midrule
    
    \multirow{2}{*}[4pt]{\shortstack[l]{Suffix Tree\\(Backward Trie)}} 
& \begin{tabular}[c]{@{}l@{}}
          Suffix Tree \\
          (Cor.~\ref{cor:dist_pal})
        \end{tabular}
      & $O(N (\timeCPincr(N, \sigma) + \min(\log h, \sigma)))$ 
      & $O(N)$ \\
\bottomrule
  \end{tabular}
\end{table}

\paragraph*{\bf Related work.}
There are a few combinatorial results for palindromes in an \emph{unrooted} edge-labeled tree.
Brlek et al.~\cite{DBLP:conf/dlt/BrlekLP15} showed
an $\Omega(M^{3/2})$ lower bound on the maximum number of distinct palindromes in an unrooted tree with $M$ edges.
Later Gawrychowski et al.~\cite{dpalforunrootedtree} showed
a matching upper bound $O(M^{3/2})$ on the maximum number of distinct palindromes in an unrooted tree with $M$ edges.
Also, they showed an $O(M^{3/2} \log M)$-time algorithm for reporting all distinct palindromes in an unrooted tree.
Note that these previous studies consider an unrooted tree.
Also, in addition to palindrome computation,
algorithms for computing some string regularities in rooted and unrooted trees are proposed~\cite{SugaharaNIBT21,DBLP:conf/spire/FiciG19,DBLP:conf/spire/RadoszewskiRSWZ21}.

The problem of constructing the suffix tree of a trie has been studied in several works.  
Kosaraju~\cite{Kosaraju89a} proposed a breadth-first construction method running in $O(N \log N)$ time,  
although the proof details were omitted and it is unclear whether the method can be extended to the online setting.  
Breslauer~\cite{Breslauer1998} presented an $O(N\sigma)$-time and space algorithm (in the proof of Theorem~2.5)  
and also proposed an $O(N \log \sigma)$-time and space algorithm obtained by encoding each character into a binary string of length $\lceil \log_2 \sigma \rceil$ and constructing the suffix tree of the binary-encoded trie.  
Although not explicitly stated, Breslauer's algorithms can be adapted for online construction.  

Our method follows a different approach, achieving $O(N)$ space while supporting online construction.  
These tasks are closely related to the problem of \emph{real-time construction},  
since online construction of index structures on tries may invalidate the amortized analyses that are valid for single-string structures.  
Real-time algorithms for constructing the suffix tree of a single string have been developed in previous work~\cite{FischerG15_AlphabetDependent,KucherovN17_FullFledged},  
and our construction algorithm generalizes these results to the trie setting.

\paragraph*{\bf Paper Organization.}
The rest of this paper is organized as follows.
Section~\ref{sec:preliminaries} provides definitions and basic properties of palindromes in a string and a trie.
Section~\ref{sec:online_maximal} introduces our algorithm for computing maximal palindromes in a trie, which utilizes periodicity of palindromes.
Section~\ref{sec:online_distinct} presents online algorithms for computing distinct palindromes in a trie.
The section describes EERTREE-based algorithms as well as a suffix-tree-based algorithm.
Finally, Section~\ref{sec:conclusions} concludes the paper.

A part of results of this article appeared in a preliminary version~\cite{FunakoshiPSC2019}.
The current article contains algorithms for computing distinct palindromes and
suffix trees/EERTREEs for a trie under leaf insertions and deletions as new materials,
which can be found in Section~\ref{sec:online_distinct}.
 \section{Preliminaries}\label{sec:preliminaries}

\subsection{Strings}

Let $\Sigma$ be the {\em alphabet}.
An element of $\Sigma$ is called a {\em character}.
An element of $\Sigma^*$ is called a {\em string}.
The length of a string $T$ is denoted by $|T|$.
The empty string $\varepsilon$ is a string of length 0,
namely, $|\varepsilon| = 0$.
For a string $T = xyz$, $x$, $y$ and $z$ are called
a \emph{prefix}, \emph{substring}, and \emph{suffix} of $T$, respectively.
For two strings $X$ and $Y$,
For a string $T$ and an integer $i$ with $1 \leq i \leq |T|$,
$T[i]$ denotes the $i$th character of $T$,
and for two integers $i,j$ with $1 \leq i \leq j \leq |T|$,
$T[i..j]$ denotes the substring of $T$
that begins at position $i$ and ends at position $j$.
For convenience, let $T[i..j] = \varepsilon$ when $i > j$.
An integer $p \geq 1$ is said to be a \emph{period}
of a string $T$ iff $T[i] = T[i+p]$ for all $1 \leq i \leq |T|-p$.

Let $\rev{T}$ denote the reversed string of $T$,
i.e., $\rev{T} = T[|T|] \cdots T[1]$.
A string $T$ is a \emph{palindrome} if $T = \rev{T}$.
We remark that the empty string $\varepsilon$ is also
considered to be a palindrome.
For any non-empty palindromic substring $T[i..j]$ in $T$,
the rational number $\frac{i+j}{2}$ is called its \emph{center}.
A non-empty palindromic substring $T[i..j]$
is said to be a \emph{maximal palindrome} centered at $\frac{i+j}{2}$ in $T$
if $T[i-1] \neq T[j+1]$, $i = 1$, or $j = |T|$.
It is clear that for each center $c = 1, 1.5, \ldots, n-0.5, n$,
we can identify the maximal palindrome $T[i..j]$ whose center is $c$
(namely, $c = \frac{i+j}{2}$).
Thus, there are exactly $2n-1$ maximal palindromes in a string of length $n$.
In particular, maximal palindromes $T[1..i]$ and $T[i..|T|]$ for $1 \leq i \leq n$
are respectively called a \emph{palindromic prefix} and a \emph{palindromic suffix} of $T$.

Manacher~\cite{Manacher75} showed an elegant online algorithm
which computes all maximal palindromes of a given string $T$ of length $n$
in $O(n)$ time.
An alternative offline approach is to use outward longest common extension (LCE) queries
for $2n-1$ pairs of positions in $T$.
Using the suffix tree~\cite{Weiner73} for string $T\$\rev{T}\#$
enhanced with a lowest common ancestor data structure~\cite{HarelT84,SchieberV88,DBLP:conf/latin/BenderF00},
where $\$$ and $\#$ are special characters which do not appear in $T$,
each outward LCE query can be answered in $O(1)$ time.
For any linearly-sortable alphabet,
preprocessing for this approach takes $O(n)$ time and space~\cite{Farach-ColtonFM00,gusfield97:_algor_strin_trees_sequen}.

Let $T$ be a string of length $n$.
For each $1 \leq i \leq n$, let $\MaxPalE_T(i)$ denote the set of
maximal palindromes of $T$ that end at position $i$.
Let $\mathbf{S}_i = s_1, \ldots, s_{g}$ be the sequence of
lengths of maximal palindromes in $\MaxPalE_T(i)$ sorted in increasing order,
where $g = |\MaxPalE_T(i)|$.
Let $d_j$ be the progression difference for $s_j$,
i.e., $d_j = s_{j} - s_{j-1}$ for $2 \leq j \leq g$.
In particular, let $d_1 = s_1 - |\varepsilon| = s_1$.
We use the following lemma which is based on
periodic properties of maximal palindromes
ending at the same position.

\begin{lemma}[Lemma 2 of \cite{DBLP:journals/tcs/FunakoshiNIBT21}]
  \label{lem:maximal_palindromes}
  \hfill
  \begin{enumerate}
    \item[(i)] For any $1 \leq j < g$, $d_{j+1} \geq d_{j}$.
    \item[(ii)] For any $1 < j < g$, if $d_{j+1} \neq d_{j}$, then $d_{j+1} \geq d_j + d_{j-1}$.
    \item[(iii)] $\mathbf{S}_i$ can be represented by $O(\log i)$ arithmetic progressions,
      where each arithmetic progression is a tuple $\langle s, d, t \rangle$ representing the sequence $s, s+d, \ldots, s + (t-1)d$ with common difference $d$.
    \item[(iv)] If $t \geq 2$, then the common difference $d$ is a period of every maximal palindrome
      which ends at position $i$ in $T$ and whose length belongs to the arithmetic progression $\langle s, d, t \rangle$.
  \end{enumerate}
\end{lemma}
Each arithmetic progression $\langle s, d, t \rangle$ is called
a \emph{group} of maximal palindromes.
It is known that all maximal palindromes in the same group share the same
preceding character in $T$, i.e., the character immediately to the left of the
palindrome.
Moreover, maximal palindromes belonging to different groups have distinct
preceding characters~\cite{NagashitaI23_PalFMIndex},
Thus, the following lemma holds.
\begin{lemma}\label{lem:num_groups}
For any position $i$, the number of groups of maximal palindromes
ending at position $i$ is bounded by $O(\min(\log i, \sigma))$.
\end{lemma}
See also Fig.~\ref{fig:arithmetic_progressions} for a concrete example.

\begin{figure}[tb]
  \centerline{
    \includegraphics[width=1.0\linewidth]{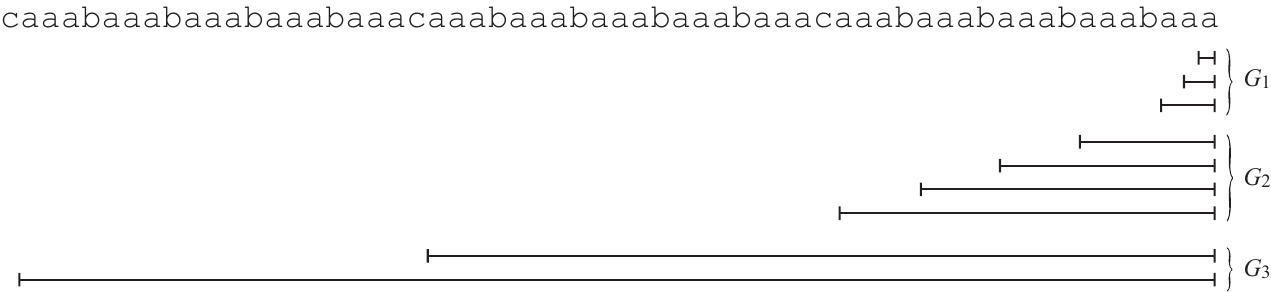}
  }
  \caption{Examples of arithmetic progressions representing
    the palindromic suffixes of a string.
    The first group $G_1$ is represented by $\langle 1, 1, 3 \rangle$,
    the second group $G_2$ by $\langle 7, 4, 4 \rangle$,
  and the third group $G_3$ by $\langle 39, 20, 2 \rangle$.}
  \label{fig:arithmetic_progressions}
\end{figure}

Since each arithmetic progression can be stored in $O(1)$ space,
and since there are only $O(\min(\log i, \sigma))$ arithmetic progressions
for each position $i$,
we can represent all maximal palindromes ending at position $i$
in $O(\min(\log i, \sigma))$ space.

For all $1 \leq i \leq n$ we can compute $\MaxPalE_T(i)$ in total $O(n)$ time:
After computing all maximal palindromes of $T$ in $O(n)$ time,
we can bucket sort all the maximal palindromes with their ending positions in $O(n)$ time.

Since palindromic suffixes are also maximal palindromes,
$\MaxPalE_T(n)$ is the set of palindromic suffixes of $T$, where $n = |T|$.
Thus Lemma~\ref{lem:maximal_palindromes} holds for palindromic suffixes in $T$.
This particular case of Lemma~\ref{lem:maximal_palindromes}
was shown in the literature~\cite{Apostolico1995parallel,MatsubaraIISNH09}.

\subsection{Maximal and Distinct Palindromes in a Trie}

A \emph{trie} $\mathcal{T} = (V, E)$ is a rooted tree whose edges are labeled by
characters from $\Sigma$ so that the out-going edges of every node carry
pairwise distinct labels.
Let $r$ be the root of $\mathcal{T}$.
When the path strings are read from the root toward the leaves we obtain the
\emph{forward} trie, and reading them in the opposite direction gives the
\emph{backward} trie.
Throughout this subsection we define notation with respect to the forward trie;
the backward version reuses the same notation after reversing the edge
directions.
Assuming the forward orientation, for any non-root node $u$ we write
$\parent(u)$ for the neighbor of $u$ that is closer to the root.
For any pair of nodes $u, v$ such that there is a directed path from $u$ to $v$
in the forward trie, we denote by $\str(u, v)$ the concatenation of edge labels
along that path.

Consider a forward trie $\mathcal{T}$ with $N$ edges.
For any palindromic path string $P = \str(u, v)$, we encode $P$ by the pair
$(|P|, v)$.
Since the reversed path from $v$ to $u$ is unique in a forward trie, the pair
$(|P|, v)$ allows us to reconstruct $P$ in $O(|P|)$ time.
A palindrome $P = \str(u, v)$ is called a \emph{maximal palindrome} in $\mathcal{T}$
if any one of the following conditions holds:
(1) $u \ne v$ and $u$ is the root,
(2) $u \ne v$ and $v$ is a leaf, or
(3) $\str(\parent(u), v')$ is not a palindrome with \emph{any} child $v'$ of $v$.

\begin{lemma} \label{lem:num_maximal_pal_trie}
  There are exactly $2N-L$ maximal palindromes in any trie $\mathcal{T}$
  with $N$ edges and $L$ leaves.
\end{lemma}

\begin{proof}
  Let $r$ be the root of $\mathcal{T}$ and $u$ be any internal node of $\mathcal{T}$.
  Since the reversed path from $u$ to $r$ is unique
  and the out-going edges of $u$ are labeled by pairwise distinct characters,
  there is a unique longest palindrome of even length (or length zero) that is centered at $u$.
  Since there are $N+1$ nodes in $\mathcal{T}$, there are exactly
  $(N+1)-L-1 = N-L$ maximal palindromes of even length in $\mathcal{T}$.

  Let $e = (u, v)$ be any edge in $\mathcal{T}$.
  From the same argument as above,
  there is a unique longest palindrome of odd length that is centered at $e$.
  Thus there are exactly
  $N$ maximal palindromes of odd length in $\mathcal{T}$.
\end{proof}

For any trie $\mathcal{T}$, let the set of all palindromic substrings
in $\mathcal{T}$ be called \emph{distinct palindromes} in $\mathcal{T}$.
See Fig.~\ref{fig:palindromes_of_a_trie} for examples
of maximal palindromes and distinct palindromes in a trie.
\begin{figure}[tb]
  \centerline{
    \includegraphics[scale=0.6]{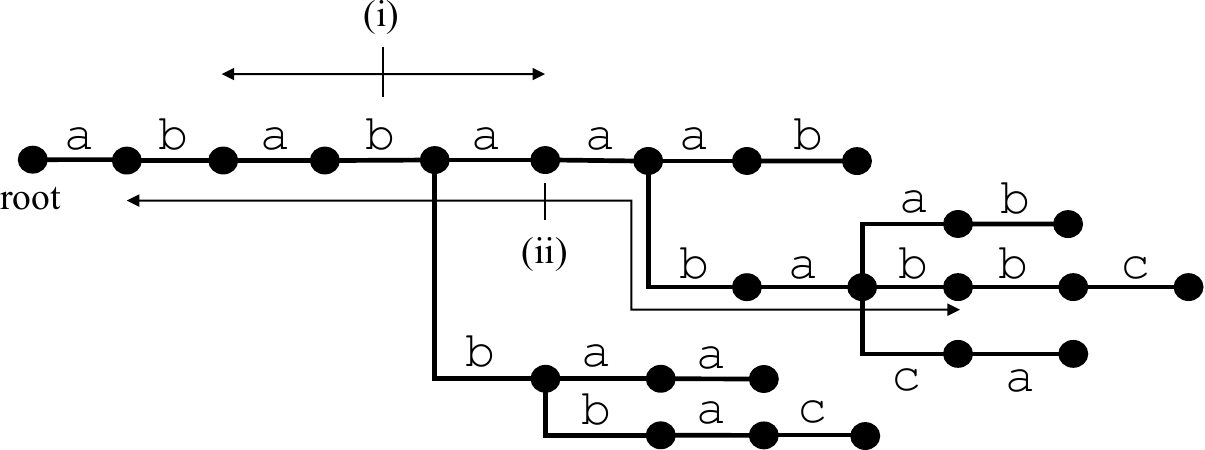}
  }
  \caption{
    The maximal palindrome centered at (i) is $\mathtt{aba}$ and
    the maximal palindrome centered at (ii) is $\mathtt{babaabab}$.
    The set of distinct palindromes in this trie is
    $\{ \mathtt{\varepsilon, a, b, c, aa, bb, aaa, aba, aca, bab}$, $\mathtt{bbb, abba, baab, aabaa, ababa, abbba, baaab, abaaba, baabaab, babaabab} \}$.
  }
  \label{fig:palindromes_of_a_trie}
\end{figure}

\begin{lemma} \label{lem:num_distinct_pal_trie}
  There are at most $N+1$ distinct palindromes in any trie $\mathcal{T}$
  with $N$ edges.
\end{lemma}

\begin{proof}
  We follow the proof from~\cite{DBLP:journals/tcs/DroubayJP01}
  which shows that the number of distinct palindromes
  in a string of length $n$ is at most $n+1$.

  We consider a top-down traversal on $\mathcal{T}$.
  The proof works with any top-down traversal but for consistency with
  our algorithm to follow, let us consider a breadth first traversal.
  Let $r$ be the root of $\mathcal{T}$
  and let $\mathcal{T}_0$ be the trie consisting only of the root $r$.
  For each $1 \leq i \leq n$,
  let $e_i  = (u_i, v_i)$ denote the $i$th visited edge in the traversal,
  and let $\mathcal{T}_i$
  denote the subgraph of $\mathcal{T}$
  consisting of the already visited edges when we have just arrived at $e_i$.
  Since we have just added $e_i$ to $\mathcal{T}_{i-1}$,
  it suffices to consider only palindromic suffixes of $\str(r, v_i)$
  since every other palindrome in $\str(r, v_i)$ already
  appeared in $\mathcal{T}_{i-1}$.
  Moreover, only the longest palindromic suffix $S_i$ of $\str(r, v_i)$
  can be a new palindrome in $\mathcal{T}_i$ which does not exist in
  $\mathcal{T}_{i-1}$,
  since any shorter palindromic suffix $S'$ is a suffix of $S_i$
  and hence is a prefix of $S_i$,
  which appears in $\mathcal{T}_{i-1}$.
  Thus there can be at most $N+1$ distinct palindromes in $\mathcal{T}$
  (including the empty string).
\end{proof}

In the following sections,
we present our algorithms for computing maximal/distinct palindromes
from a given trie $\mathcal{T}$.
We denote by $N$ the number of edges in $\mathcal{T}$,
by $D \le N$ the number of distinct palindromes in $\mathcal{T}$,
by $h \le N$ the height of $\mathcal{T}$, and
by $L \le N$ the number of leaves in $\mathcal{T}$.

\subsection{Predecessor Data Structures}~\label{sec:colored_pred}
We first recall the standard predecessor problem.
Let $\mathcal{L}$ be an ordered list.
The predecessor query $\Pred(e)$ asks for the largest element $e' \in \mathcal{L}$
such that $e' < e$.
A predecessor data structure supports insertions, deletions,
and predecessor queries on $\mathcal{L}$.
We denote by $\timeP(x) \in O(\log x)$ the query time of a predecessor data structure
that supports all operations on a set of size $O(x)$ over a universe of size $O(x)$,
using linear space with respect to the size of the set.
We also denote by $\timePincr(x)$ the query time of the incremental variant
that supports only insertions and predecessor queries.
A comprehensive overview of existing predecessor data structures can be found in~\cite{NavarroR20_PredecessorSearch}.

In this paper, predecessor data structures are mainly used as dictionaries
for the outgoing edges of trie nodes, supporting access to a child node
given an edge label.
Since the outgoing edges of each trie node are labeled by characters from an alphabet of size $\sigma$,
the cost of dictionary operations on these edges is expressed as
$\timeP(\sigma)$ in the fully dynamic case and as $\timePincr(\sigma)$
when only insertions are supported.

In Section \ref{sec:online_distinct}, we use \emph{colored predecessor} data structures~\cite{Mortensen06FullyDynamicOnRAM,Mortensen03FullyDynamicTwoDimensional}.
Let $\mathcal{L}$ be an ordered list and $C$ be a set of \emph{colors}.
Each element $e \in \mathcal{L}$ is associated with a set of colors $C(e) \subseteq C$.
We define the following six operations and queries.
\begin{itemize}
\item $\Insert(e)$: add a new element immediately after $e$.
\item $\Delete(e)$: delete the element $e$ from $\mathcal{L}$.
\item $\Color(e, c)$: add a color $c$ to $C(e)$.
\item $\Uncolor(e, c)$: remove a color $c$ from $C(e)$.
\item $\Pred(e, c)$: find the closest element $e' \in \mathcal{L}$ such that $e' < e$ and $c \in C(e')$.
\item $\Succ(e, c)$: find the closest element $e' \in \mathcal{L}$ such that $e' > e$ and $c \in C(e')$.
\end{itemize}
A colored predecessor data structure is a data structure that supports these queries.
In the following, we consider the case where the set of colors $C$ is the alphabet $\Sigma$, i.e., each element $e \in \mathcal{L}$ is associated with a set of characters $C(e) \subseteq \Sigma$.
We also assume that the alphabet $\Sigma$ is ordered.

A natural approach to implementing the data structure for a general ordered alphabet is to maintain a predecessor data structure for each color.  
If balanced binary search trees are used as the predecessor data structures, and another binary search tree is maintained to locate the predecessor data structure associated with each character,  
the overall data structure supports queries in $O(\log m)$ time and requires $O(m)$ space in total.
In addition to this straightforward method, we employ the existing data structures for the colored predecessor problem.
\begin{lemma}\label{lem:colored_pred}
There exist colored predecessor data structures such that
\begin{enumerate}
\item supporting $\Insert(e)$, $\Color(e, c)$, $\Pred(e, c)$, and $\Succ(e, c)$ in $O(\log \log m)$ time when $\sigma \leq \log^{1/4} m$ and the alphabet is an integer range $[1,\sigma]$~\cite{Mortensen03FullyDynamicTwoDimensional,KucherovN17_FullFledged},
\item supporting all operations in $O\left(\frac{(\log \log m)^2}{\log \log \log m}\right)$ time when $\sigma \in O(n)$ and the alphabet is an integer range $[1,\sigma]$~\cite{Mortensen06FullyDynamicOnRAM}, and
\item supporting all operations in $O(\log m)$ time for a general ordered alphabet
\end{enumerate}
using $O(m)$ space, where $m$ is the maximum total number of colors associated with any element.
\end{lemma}
We denote by $\timeCP(x, \sigma)$ the query time of the linear-space colored predecessor data structure with $\sigma$ colors that supports all six operations, where the maximum total number of elements and colors associated with any element is $O(x)$, and by $\timeCPincr(x, \sigma)$ the query time of the incremental variant that supports only $\Insert$, $\Color$, $\Pred$, and $\Succ$ under the same conditions.
Since the colored predecessor problem generalizes the standard predecessor problem,
we have $\timeP(x) \in O(\timeCP(x, \sigma))$ and
$\timePincr(x) \in O(\timeCPincr(x, \sigma))$.
 \section{Online Computation of Maximal Palindromes in a Trie}\label{sec:online_maximal}
In this section, we present an online algorithm for computing maximal palindromes in a given trie $\mathcal{T}$.
Our algorithm runs in $O(N \min(\log h, \sigma))$ time and $O(N)$ space.
In this section, we treat the input trie as a \emph{forward trie}, i.e., we read path-labels in the root-to-leaf direction.
We note that the algorithm in this section works over a general \emph{unordered} alphabet.

The basic strategy of our algorithm is as follows:
For each node $u$,
we maintain the arithmetic progression
that represents the maximal palindromes ending at $u$
sorted in the increasing order of their lengths.
By Lemma~\ref{lem:num_maximal_pal_trie},
the number of arithmetic progressions to store is at most $2N-L$ in total.
Suppose that a node $u$ has a child $v$ with edge-label $c_v$
and that a new leaf $\ell$ with edge-label $a \ne c_v$ is added as a child of node $u$.
Notice that some of the maximal palindromes ending at $u$ could be extended by character $a$.
Since $a \ne c_v$,
those palindromes that are not extended with $a$ could still be extended with $c_v$.
This in turn means that
when another leaf $\ell'$ is added as a child of $u$ later,
then we can use the maximal palindromes that end at $u$
for finding the palindromes ending at the new leaf $\ell'$.
In the sequel, we will describe how to efficiently maintain these maximal palindromes
for each additional leaf of the trie.

Suppose that now we are to add a leaf $\ell$ with edge-label $a$ as a child of $u$.
The task here is to check if each palindromic suffix ending at $u$ can be extended with the character $a$.
We will process the groups of palindromic suffixes ending at $u$ in increasing order of their lengths.
Let $\langle s, d, t \rangle$ be the arithmetic progression
representing a given group of palindromic suffixes ending at $u$,
where $s$ is the length of the shortest palindromic suffix
in the group, $d$ is a common period of the palindromic suffixes
and $t$ is the number of palindromic suffixes in this group.
The cases where $t = 1$ and $t = 2$ are trivial,
so we consider the case where $t \geq 3$.
Let $P$ be a palindromic suffix in the group
that is not the longest one (i.e., $s \leq |P| \leq s + (t-2)d$).
Due to the periodicity (Claim~(iv) of Lemma~\ref{lem:maximal_palindromes}),
every $P$ is immediately preceded by a unique string $P[1..d]$ of length $d$.
Let $b = P[d]$ and $c$ be the character that immediately precedes
the longest palindromic suffix in the group.
There are four cases to consider:
\begin{enumerate}
  \item \label{case:1} $a = b$ and $a = c$ (namely $a = b = c$):
    In this case, all the palindromic suffixes in the group extend with $a$ and become
    palindromic suffixes of $\str(r, \ell)$.
    We update $s \leftarrow s+2$.
    The values of $d$ and $t$ stay unchanged.
  \item \label{case:2} $a = b$ and $a \neq c$.
    In this case, all the palindromic suffixes but the longest one in the group
    extend with $a$ and become palindromic suffixes of $\str(r, \ell)$.
    We update $s \leftarrow s + 2$ and $t \leftarrow t-1$.
    The value of $d$ stays unchanged.
\item \label{case:3} $a \neq b$ and $a = c$.
    In this case, only the longest palindromic suffixes in the group
    extends with $a$ and becomes a palindromic suffix of $\str(r, \ell)$.
    We first update $s \leftarrow s+(t-1)d+2$ and then $t \leftarrow 1$.
    The new value of $d$ is easily calculated from
    the length of the longest palindromic suffix in the previous group
    (recall the definition of $d$ just above Lemma~\ref{lem:maximal_palindromes}).
\item \label{case:4} $a \neq b$ and $a \neq c$.
    In this case, none of the members in the group extends with $a$.
    Then, we do nothing.
\end{enumerate}
In each of the above cases,
we store all these extended palindromes in $\ell$ as the set of maximal palindromes
ending at $\ell$ in $\str(r, \ell)$,
and exclude all these extended palindromes from the set of maximal palindromes
ending at $u$.
\begin{example}
  See Fig.~\ref{fig:arithmetic_progressions} for concrete examples of the above cases.
  Let $\alpha$ be the next character that is appended to the string
  in Fig.~\ref{fig:arithmetic_progressions}.
  Case~\ref{case:1} occurs to group $G_3$ when $\alpha = \mathtt{c}$.
  Case~\ref{case:2} occurs to group $G_1$ when $\alpha = \mathtt{a}$, and to group $G_2$ when $\alpha = \mathtt{b}$.
  Case~\ref{case:3} occurs to group $G_1$ when $\alpha = \mathtt{b}$, and to group $G_2$ when $\alpha = \mathtt{c}$.
  Case~\ref{case:4} occurs to all the groups when $\alpha = \mathtt{d}$.
\end{example}
In each of the above four cases,
we can check by at most two characters comparisons if the palindromes in a given group can be extended with $a$.
Since there are $O(\min(\log h, \sigma))$ arithmetic progressions
representing the palindromic suffixes ending at node $u$,
it takes $O(\min(\log h, \sigma))$ time to compute the palindromic suffixes ending at $\ell$.

There is one more issue remaining.
When only one or two members from a group extend with $a$,
then we may need to merge these palindromic suffixes
into a single arithmetic progression
with the palindromic suffixes from the previous group.
However, this can easily be done in a total of $O(\min(\log h, \sigma))$ time per leaf $\ell$,
since the palindromic suffixes ending at $u$ were given as $O(\min(\log h, \sigma))$ arithmetic progressions (groups).
\begin{example}
  See Fig.~\ref{fig:arithmetic_progressions} for a concrete example of the merging process.
  When $a = \mathtt{c}$,
  $\mathtt{c}$ is a palindromic suffix and forms a single arithmetic progression
  $\langle 1, 0, 1 \rangle$.
  All the palindromes in $G_1$ are not extended.
  The longest palindromic suffix in group $G_2$ is extended
  to $\mathtt{caaabaaabaaabaaabaac}$
  forming an arithmetic progression $\langle 21, 20, 1 \rangle$,
  where $20 = |\mathtt{caaabaaabaaabaaabaac}| - |\mathtt{c}|$,
  but all the other palindromic suffixes in group $G_2$ are not extended.
  Finally all the palindromic suffixes in group $G_3$ are extended
  and are represented by an arithmetic progression $\langle 41, 20, 2 \rangle$.
  Since the three palindromic suffixes of lengths $21$, $41$, and $61$
  share the common difference $20$,
  the two arithmetic progressions
  are merged into a single arithmetic progression $\langle 21, 20, 3 \rangle$.
\end{example}
We have shown the following:
\begin{theorem}\label{thm:MPalOnline}
  We can compute all maximal palindromes in a trie $\mathcal{T}$ in an online manner in $O(N \min(\log h, \sigma))$ time and $O(N)$ space.
\end{theorem}
  When a leaf $\ell$ is deleted, we purge information on $\ell$ of size $O(\min(\log h, \sigma))$ and
  execute the reverse operation of the above merging process if needed.
  We obtain the next corollary:
  \begin{corollary}\label{col:dynamicMPal}
    We can maintain the set of maximal palindromes in a trie,
    which allows inserting or deleting a leaf in a single operation,
    in the worst-case $O(\min(\log h, \sigma))$ time per one operation.
  \end{corollary}

 \section{Online Computation of Distinct Palindromes in a Trie}\label{sec:online_distinct}
In this section, we present several online algorithms for computing all distinct palindromes in a given trie $\mathcal{T}$.  
Our algorithms work on a general ordered alphabet.

\subsection{EERTREE-based Online Algorithm}
This subsection presents an algorithm for enumerating all distinct palindromes in a trie
based on an \emph{EERTREE}.
In this subsection, we treat the input trie $\mathcal{T}$ as a \emph{forward trie}.

\subsubsection{EERTREE} \label{ssse:eertree}

The definition of the EERTREE of a trie
is a natural generalization of that of a string:
The nodes of the EERTREE $\EERTREE(\mathcal{T})$ of a trie $\mathcal{T}$
correspond to the distinct palindromes in $\mathcal{T}$.
We sometimes identify each node of $\EERTREE(\mathcal{T})$ with its corresponding palindrome.
There is a special node $\bot$ that represents an imaginary string of length $-1$.
All palindromes of length $1$ are children of $\bot$.  
For convenience, we define $c \cdot \bot \cdot c = c$.
Each node $u$ with $|u| \ge 2$ is a child of the node $v$ with $v = u[2.. |P|-1]$.
As a result, $\EERTREE(\mathcal{T})$ consists of two rooted trees
(1)       one is rooted at $\bot$        and contains all  odd-length palindromes in $\mathcal{T}$, and
(2) the other is rooted at $\varepsilon$ and contains all even-length palindromes in $\mathcal{T}$.
By definition, the size of $\EERTREE(\mathcal{T})$ is linear in the number of distinct palindromes in $\mathcal{T}$, which is at most $N+1$~(Lemma~\ref{lem:num_distinct_pal_trie}).
For any string $x$, let $\LPS(x)$ denote the longest palindromic suffix of $x$, possibly equal to $x$ itself.
We maintain a pointer from each trie node representing a string $x$ to the node $\LPS(x)$ in the EERTREE.

As in the original paper,
we introduce three types of links used to efficiently perform online construction.
For each node $x$ in $\EERTREE(\mathcal{T})$, the \emph{suffix link} $\slink(x)$ is a link to the node corresponding to the longest proper palindromic suffix of $x$ (i.e., one strictly shorter than $x$).
The \emph{quick link} $\qlink(x)$ is a link to the node corresponding to
the longest proper palindromic suffix of $x$ that is preceded by a character different from $b$,
where $b$ is the character preceding the suffix $\slink(x)$ of $x$.
For each node $x$ and each character $c \in \Sigma$, the \emph{direct link} $\dlink(x, c)$ is a link to the node corresponding to
the longest proper palindromic suffix of $x$ that is preceded by $c$

For convenience, we set $\slink(\bot) = \qlink(\bot) = \dlink(\bot, c) = \bot$.
We also set $\qlink(x)$ and $\dlink(x, c)$ to $\bot$ if such palindromes do not exist.

Since $\EERTREE(\mathcal{T})$ represents the distinct palindromes in $\mathcal{T}$,
an online construction of $\EERTREE(\mathcal{T})$ is equivalent to an online enumeration of distinct palindromes in $\mathcal{T}$.
Consider the case where a new leaf $v$ is added to the trie with edge label $c$ from its parent $u$.
Then, from the proof of Lemma~\ref{lem:num_distinct_pal_trie}, there is at most one new palindrome, namely $\LPS(\str(r, v))$.
Therefore, updating $\EERTREE(\mathcal{T})$ can be done as follows:
\begin{enumerate}
  \item compute $\LPS(\str(r, v))$;
  \item insert it into $\EERTREE(\mathcal{T})$ if it is a new palindrome; and
  \item update additional information of $\EERTREE(\mathcal{T})$ used in the algorithm.
\end{enumerate}

In the following, we discuss how to efficiently compute $\LPS(\str(r, v))$
and update the additional information.
We present three algorithms for this task, namely,
the \emph{basic}, \emph{quickLink}, and \emph{directLink} algorithms.
Each of these algorithms follows the same underlying ideas
as the corresponding algorithms proposed in the original paper.

\subsubsection{The \emph{basic} Algorithm} \label{ssse:basic}
The simplest way to compute $\LPS(\str(r, v))$ is to
enumerate all palindromic suffixes of $\str(r, v)$ in decreasing order of length,
check for each palindrome $x$ whether $cxc$ is a palindromic suffix of $\LPS(\str(r, v))$, and
select the first one that satisfies the condition.
This can be achieved by following the suffix links starting from $\LPS(\str(r, u))$.
By storing, for each trie node, the character that precedes its longest palindromic suffix,
each check can be performed in constant time.

For the single-string case, the total number of traversals over all $m$ insertions is $O(m)$.
However, in the trie setting, the amortized time analysis does not apply.
Consider the trie $\mathcal{T}$ that represents the set $\{ \mathtt{a}, \mathtt{aa}, \dots, \mathtt{a}^m \}$ and its corresponding EERTREE.
Suppose that we insert $m$ new nodes representing $\mathtt{a}^m\mathtt{b}, \mathtt{a}^m\mathtt{c}, \mathtt{a}^m\mathtt{d}, \dots$.
In this case, for each insertion, the algorithm must traverse suffix links starting from the node corresponding to $\mathtt{a}^m$ down to $\bot$.
Since each traversal takes $O(m)$ steps, the total cost of all insertions becomes $O(m^2)$.

\subsubsection{The \emph{quickLink} Algorithm} \label{ssse:quicklink}

Consider an EERTREE node $x$ during the procedure of the previous algorithm.
In this case,
if the preceding character of $\slink(x)$
and that of $\slink(\slink(x))$ are the same, we can skip checking
whether $c \cdot \slink(\slink(x)) \cdot c$ is a candidate.
Thus, replacing $\slink$ with $\qlink$ allows us to skip
such unnecessary checks without affecting the correctness of the algorithm.
Specifically,
we can compute $\LPS(\str(r, v))$
by following the quick links starting from $\LPS(\str(r, u))$,
checking for each palindrome $x$ whether $cxc$ is a palindromic suffix of $\LPS(\str(r, v))$,
and selecting the first one that satisfies the condition.
To enable checking whether the preceding character of $\qlink(x)$ is $c$,
for each trie node $x$, we store the character that precedes the suffix $\qlink(x)$ in $x$.
Since $\qlink$ skips over palindromic suffixes with the same preceding character,
by Lemma~\ref{lem:num_groups},
the number of traversals using $\qlink$ is bounded by
$O(\min(\log |\str(r, v)|, \sigma)) \subseteq O(\min(\log h, \sigma))$ time.
Thus, we can compute $\LPS(\str(r, v))$ in $O(\min(\log h, \sigma))$ time.

After inserting a new node, we need to add several links for the new node $x$ in $\EERTREE(\mathcal{T})$.
We maintain $\slink$ and $\qlink$ in this algorithm.
The suffix link $\slink(x)$ corresponds to the second-longest palindromic suffix of $\str(r, v)$.
Such a palindrome can be found by traversing $\qlink$ in the same way as when determining the new node,
in $O(\log h)$ time.
By the definition of $\qlink(x)$,
$\qlink(x)$ is the longest palindromic suffix $y$ of $x$
such that the preceding character of $y$ is different from that of $\slink(x)$.
It is either $\slink(\slink(x))$ or $\slink(\qlink(x))$,
which can be computed in constant time.

The total running time for determining $\LPS(\str(r, v))$
is $O(\min(\log h, \sigma))$, and inserting a new node and updating its links take
$O(\min(\log h, \sigma) + \timePincr(\sigma))$ time,
where the $O(\min(\log h, \sigma))$ term accounts for traversals using $\qlink$
and the $O(\timePincr(\sigma))$ term arises from edge insertions
using predecessor data structures.
Since the space complexity for storing the EERTREE and $\qlink$ is $O(N)$,
and the total number of insertion of new EERTREE nodes is $O(D)$,
we obtain the following result.
\begin{theorem}\label{thm:eertree_online_qlink}
  The EERTREE for a trie $\mathcal{T}$ can be constructed in an online manner in $O(N \min(\log h, \sigma) + D \timePincr(\sigma))$ time and $O(N)$ space.
\end{theorem}
\begin{corollary} \label{cor:eertree_online_qlink}
  We can compute all distinct palindromes in a trie $\mathcal{T}$ in an online manner in $O(N \min(\log h, \sigma) + D \timePincr(\sigma))$ time and $O(N)$ space.
\end{corollary}

\subsubsection{The \emph{directLink} Algorithm} \label{ssse:directlink}

In the above algorithm,
traversals using $\slink$ or $\qlink$ are performed
when we find a node $x$ such that
the preceding character of $\slink(x)$ is $c$.
By using $\dlink$ with character $c$,
we can directly jump to such a node.
Specifically, the new palindrome is either of the form
$c \cdot \LPS(\str(r, u)) \cdot c$
or $c \cdot \slink(\dlink(\LPS(\str(r, u)), c)) \cdot c$.
In addition,
since the $\slink$ of the new node point to the second-longest palindromic suffix of $\str(r, v)$,
we can compute it by using $\dlink$ as well,
in the same way as when determining the new node.
An illustration of the procedure for inserting a new node is shown in Fig.~\ref{fig:eertree}.

\begin{figure}
  \hfill
  \begin{center}
  \begin{minipage}{0.45\textwidth}
    \centering
    \includegraphics[scale=0.7]{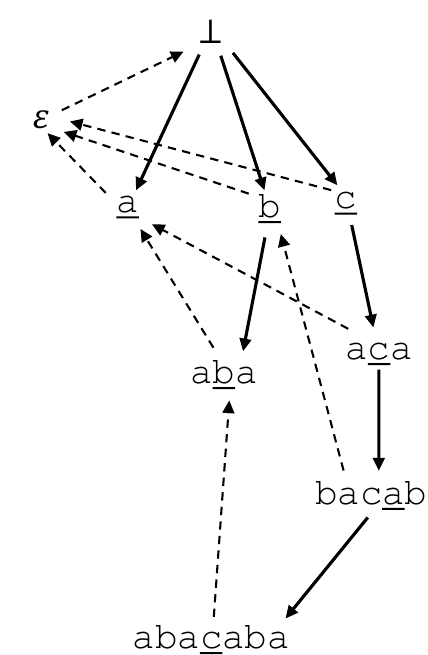}
  \end{minipage}
  \begin{minipage}{0.45\textwidth}
    \centering
    \includegraphics[scale=0.7]{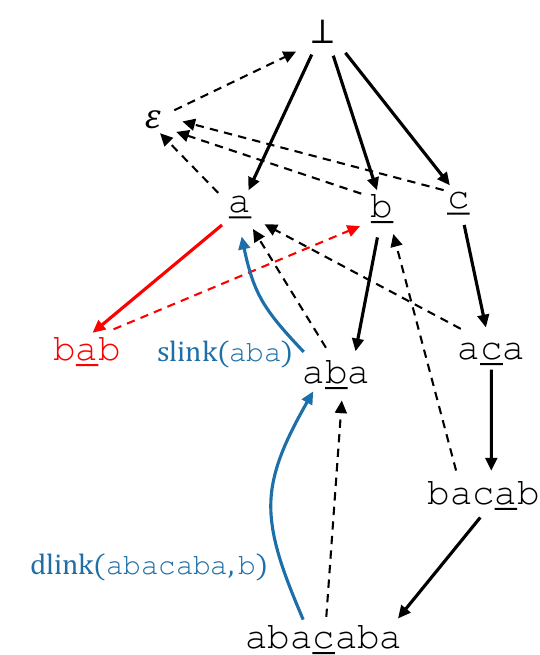}
  \end{minipage}
  \newline
  \end{center}
  \caption{
    An illustration of an EERTREE $\EERTREE(T)$ of a string $T = \mathtt{abacaba}$ (left)
    and the updated tree after 
    a new character $\mathtt{b}$ is appended to it, forming $\mathtt{abacabab}$ (right).
    The red nodes and edges represent newly created parts of $\EERTREE(T)$.
    Solid arrows indicate edges in the EERTREE, while dashed arrows represent suffix links.
    The process of locating the insertion point for the new node is shown by blue arrows.
    In this example, 
    the insertion point is determined as
    $\slink(\dlink(\LPS(\mathtt{abacaba}), \mathtt{b})) = \slink(\dlink(\mathtt{abacaba}, \mathtt{b})) = \slink(\mathtt{aba}) = \mathtt{a}$.
    Since the palindrome $\mathtt{bab}$ does not yet exist in $\EERTREE(T)$, a new node for it is inserted.
    Note that this procedure for a single string can also be applied to the trie setting,
    where appending a character corresponds to adding a new leaf.
  }
  \label{fig:eertree}
\end{figure}

The algorithm using direct links has an efficiency issue.
Since a single node may have multiple direct links, storing them in a straightforward manner results in significant time and space overhead.
Specifically, each node may have up to $\sigma$ direct links; thus, storing all direct links for each node requires $O(N \sigma)$ space.
To address this issue, we use two types of methods:
the \emph{Persistent Tree} method and
the \emph{Colored Ancestor} method.

\paragraph*{\bf The \emph{Persistent Tree} Method}
The original paper~\cite{EERTREE} employed a fully persistent binary search tree
(\emph{persistent tree} for short)~\cite{DriscollSST86_Persistent}.
A persistent tree stores an ordered sequence.
It allows
copying a tree in $O(\log k)$ time and space,
inserting an element in $O(\log k)$ time and space,
deleting an element in $O(\log k)$ time and space, and
accessing an element in $O(\log k)$ time and constant space,
where $k$ is the number of nodes in the tree.
We directly apply this technique to the EERTREE of a trie.

In the algorithm, the direct links associated with each node are represented by a persistent tree.
By Lemma~\ref{lem:num_groups}, the number of groups of palindromic suffixes
associated with a node is $O(\min(\log h, \sigma))$.
Thus, the number of direct links is bounded by $O(\min(\log h, \sigma))$
and storing these links in a persistent tree allows us to access $\dlink(x, c)$
in $O(\log \min(\log h, \sigma))$ time.

When a new node $x$ is added to the EERTREE,
we create a new persistent tree representing the direct links of $x$.
Since each direct link points to a node that can be reached by repeatedly following $\slink$ from $x$,
the direct links of node $x$ are almost the same as those of node $\slink(x)$.
Specifically, the value of $\dlink(x, a)$ can be computed as follows.
\begin{align*}
  \dlink(x, a) = \begin{cases}
    \slink(x)            & (\text{$x$ ends with $a \cdot \slink(x)$}) \\
    \dlink(\slink(x), a) & (\text{otherwise})
  \end{cases}
\end{align*}
Thus, the direct links of a node $x$ differ from those of $\slink(x)$ only for the character
preceding the suffix $\slink(x)$ in $x$.
Consequently, the persistent tree for $x$ can be obtained by copying that of $\slink(x)$
and updating the corresponding entry to point to $\slink(x)$,
which requires only a constant number of insertions and deletions.
This operation takes $O(\log \min(\log h, \sigma))$ time and space.

For each trie node insertion,
we spend $O(\log \min(\log h, \sigma))$ time accessing direct links
and determining whether a new palindrome should be inserted.
If a new palindrome is created,
we take $O(\log \min(\log h, \sigma) + \timePincr(\sigma))$ time
and $O(\log \min(\log h, \sigma))$ space
to insert it and update its direct links.
Therefore, we obtain the following result.
\begin{theorem}\label{thm:eertree_online_dlink_persistent}
  The EERTREE for a trie $\mathcal{T}$ can be constructed in an online manner in 
  $O(N \log \min(\log h, \sigma) + D \timePincr(\sigma))$ 
  time and $O(N + D \log \min(\log h, \sigma))$ space.
\end{theorem}
\begin{corollary} \label{cor:eertree_online_dlink_persistent}
  We can compute all distinct palindromes in a trie $\mathcal{T}$ in an online manner in 
  $O(N \log \min(\log h, \sigma) + D \timePincr(\sigma))$ 
  time and $O(N + D \log \min(\log h, \sigma))$ space.
\end{corollary}

\paragraph*{\bf The \emph{Colored Ancestor} Method}
The \emph{Persistent Tree} method works well for the EERTREE for a trie, but it does not guarantee linear space usage.
In this work, we present an alternative method for implementing direct links in a more space-efficient way.

For an EERTREE $\EERTREE(\mathcal{T}) = \langle V, E \rangle$,
the suffix link tree of $\EERTREE(\mathcal{T})$
is a tree where the vertex set is $V$ and
the edge set is $\{ (x, \slink(x)) \mid x \in V \}$.
Since the path from a node $x$ to $\bot$ in the suffix link tree represents the sequence of all palindromic suffixes of $x$ in decreasing order of length,
the direct link $\dlink(x, c)$ points to the nearest ancestor $y$ of $x$ such that $c \cdot \slink(y)$ is a suffix of $y$.
Therefore, if we regard each node $x$ as being colored by the character that precedes $\slink(x)$,
the direct link $\dlink(x, c)$ corresponds to the \emph{nearest colored ancestor (NCA)}~\cite{GawrychowskiLMW18_NearestColoredNode} with color $c$ in the suffix link tree.

The NCA problem we address is in the \emph{incremental} and \emph{semi-dynamic} setting.
In our case, new nodes are inserted only as leaves due to the structure of the suffix link tree.
Moreover, once a node is created in the suffix link tree, its color remains unchanged.
Therefore, the only dynamic operations that our NCA data structure needs to support is the leaf insertion and NCA queries.
We summarize our semi-dynamic NCA problem as follows.

\begin{definition} \label{def:semi_dynamic_LCA}
The \emph{semi-dynamic nearest colored ancestor (NCA) problem} maintains a rooted tree whose nodes are colored from an alphabet $\Sigma$ and supports the following operations:
\begin{itemize}
  \item $\InsertLeaf(u, c)$: insert a new leaf as a child of $u$ and assign color $c$.
  \item $\DeleteLeaf(v)$: delete a leaf $v$ from the tree.
  \item $\NCA(v, c)$: return the nearest ancestor of $v$ that colored by $c$, or $\bot$ if no such node exists.
\end{itemize}
The problem is called the \emph{incremental semi-dynamic NCA problem} when the $\DeleteLeaf$ operation is not supported.
\end{definition}

We first describe the static NCA data structure proposed in~\cite{GawrychowskiLMW18_NearestColoredNode}.
The data structure maintains a static predecessor data structure for each color $c$, storing all nodes colored $c$ in Euler tour order.
An Euler tour of a tree with $k$ nodes is a sequence of nodes of length $2k$
obtained by performing a depth-first traversal of the tree and listing each node twice,
the first time it is visited and the last time before returning from it.
Thus, each node colored $c$ appears exactly twice in the predecessor data structure for color $c$.
Additionally, for each node $x$ marked with color $c$, it stores a pointer to the nearest colored ancestor of $x$ (excluding $x$ itself) with color $c$.
To find the NCA of a node $x$ with color $c$, we perform a predecessor query to locate the closest occurrence of a node colored $c$ that appears before the first occurrence of $x$ (i.e., its preorder visit) in the Euler tour order.
If this occurrence corresponds to the first appearance of some node $y$, then $y$ is the NCA of $x$ with respect to color $c$.
Otherwise, it corresponds to the second appearance (i.e., the postorder visit) of some node $y$.
In this case, the NCA of $x$ with color $c$ is the same as the NCA of $y$ with color $c$ excluding $y$ itself.
Since $y$ stores a pointer to such a node, we can compute the NCA of $x$ efficiently.

Then, we show how to update the data structure for supporting $\InsertLeaf$.
The static NCA data structure consists of two parts: 
predecessor data structure for each color, and 
pointers to the nearest colored ancestor.
For color-wise predecessor data structures,
we can simply replace them with an incremental colored predecessor data structure described in Lemma~\ref{lem:colored_pred}.
Since our semi-dynamic NCA problem does not allow changing the node color and inserting an internal node,
once a node is inserted, its nearest colored ancestor does not change.
Therefore, we do not have to update the pointers after new node is inserted.
Once new node $x$ is inserted with label $c$,
we compute $\NCA(x, c)$ and store the pointer to it.

We can also support $\DeleteLeaf$ queries.
To delete the element from the colored predecessor data structure,
we replace an incremental colored predecessor data structure with
the data structure that also supports deletions described in Lemma~\ref{lem:colored_pred}.
Furthermore, the pointers to the precomputed NCAs can be updated by simply removing those of the deleted node.
It is because the deleted node in $\EERTREE(\mathcal{T})$ is always a leaf and
there is no need to update other NCA pointers.  

Each update and query in the semi-dynamic NCA data structure can be reduced to a constant number of
updates and queries to the colored predecessor data structures defined over the Euler tour order.
Therefore, the time complexity of each operation is determined by that of the underlying colored
predecessor structure.
We recall that $\timeCPincr$ and $\timeCP$ denote the query times of
incremental and fully dynamic colored predecessor data structures,
respectively, as defined at the end of Section~\ref{sec:colored_pred}.
The space complexity remains linear in the number of nodes maintained.
This is because each node is assigned at most one color, and hence appears in at most one
colored predecessor data structure and stores only a single NCA pointer.
This leads to the following lemma.
\begin{lemma}\label{lem:semi_dynamic_nca}
  For a tree of size $m$,
  the incremental semi-dynamic nearest colored ancestor problem can be supported in
  $O(\timeCPincr(m, \sigma))$ time per operation using $O(m)$ space.
  We can also support delete leaf operations
  in $O(\timeCP(m, \sigma))$ time per operation using $O(m)$ space.
\end{lemma}

We now apply the above result to the suffix link tree of $\EERTREE(\mathcal{T})$ in order to implement direct links.
During the online construction of $\EERTREE(\mathcal{T})$, each insertion of a new trie node induces
only a constant number of operations on the semi-dynamic NCA data structure, corresponding to the
computation and maintenance of direct links.
In addition, we can insert a new node in $O(\timePincr(\sigma)) \subseteq O(\timeCPincr(D, \sigma))$ time using predecessor data structures.
Consequently, the time required for processing each trie node is
$O(\timeCPincr(D, \sigma))$, and the total space usage remains linear in $N$.
From this, we obtain the following theorem and corollary.

\begin{theorem} \label{thm:eertree_online_dlink_nca}
  The EERTREE for a trie $\mathcal{T}$ can be constructed in an online manner in $O(N (\timeCPincr(D, \sigma)))$ time and $O(N)$ space.
\end{theorem}
\begin{corollary} \label{cor:eertree_online_dlink_nca}
  We can compute all distinct palindromes in a trie $\mathcal{T}$ in an online manner in $O(N (\timeCPincr(D, \sigma)))$ time and $O(N)$ space.
\end{corollary}

\subsubsection{Supporting Leaf Deletions} \label{ssse:leaf_deletion}
Finally, we show that the EERTREE of a trie can be modified to support leaf deletions.

When a leaf is removed from the trie, we use the following procedure to decide which node in $\EERTREE(\mathcal{T})$ should also be deleted:
Let $x$ be the string corresponding to the deleted leaf, and consider the node $\LPS(x)$ in $\EERTREE(\mathcal{T})$.  
If $\LPS(x)$ is a leaf in the EERTREE and has exactly one incoming link from $\mathcal{T}$, then it should be deleted.  
Otherwise, only the trie node is removed.

We now argue the correctness of this procedure.
If $\LPS(x)$ is not a leaf of the EERTREE, then $\LPS(x)$ occurs elsewhere, so it cannot be removed.
If there are multiple links to $\LPS(x)$, then there exists another trie node whose longest palindromic suffix is also $\LPS(x)$, and again it must be retained.
Finally, for the sake of contradiction, assume that
$\LPS(x)$ is a leaf, that there is exactly one link to this leaf from $\mathcal{T}$,
and that $\LPS(x)$ occurs in $\mathcal{T}$ other than as a suffix of $x$.
Let $y$ be the path string of the trie node with the smallest depth such that $\LPS(x)$ occurs as a suffix of its root-to-node path string.
Since there is only one link to $\LPS(x)$, the palindrome $\LPS(x)$ cannot be the longest palindromic suffix of $y$.
Thus, $\LPS(x)$ must occur as a proper prefix of the longest palindromic suffix of $y$,
which contradicts the minimality of the depth of $y$.

Once the EERTREE node to be deleted has been identified, we also update the additional data structures used to simulate links.
The deletion of $\slink$ and $\qlink$ is straightforward.
For $\dlink$, we use either a persistent tree or an NCA data structure.
In the case of the persistent-tree method,
when a node $x$ is deleted,
we delete the persistent tree representing the direct links of $x$.
This does not affect other persistent trees.
This is because each persistent tree of a node $y$ is copied from that of $\slink(y)$,
and the deleted node must be a leaf in the suffix-link tree.
In the case of the NCA-based method,
we can simply replace the incremental semi-dynamic NCA data structure
with the semi-dynamic NCA data structure described in Lemma~\ref{lem:semi_dynamic_nca}.
Therefore, by combining the above results with the leaf-deletion procedures described above,
we can support leaf deletions within the following time bounds.

\begin{corollary}\label{cor:dpal_deletion}
  The set of distinct palindromes in a trie $\mathcal{T}$,
  which supports both leaf insertions and leaf deletions,
  can be maintained online with the following time and space bounds:
  \begin{itemize}
    \item
    Using the Quick Link method
    (Corollary~\ref{cor:eertree_online_qlink}),
    each update can be supported in
    $O(\log \min(\log h, \sigma) + \timeP(\sigma))$ time and $O(N)$ space.
    \item
    Using the Direct Link with Persistent Tree method
    (Corollary~\ref{cor:eertree_online_dlink_persistent}),
    each update can be supported in
    $O(\log \min(\log h, \sigma) + \timeP(\sigma))$ 
    time and
    $O(N + D \log \min(\log h, \sigma))$ space.
    \item
    Using the Direct Link with Colored Ancestor method
    (Corollary~\ref{cor:eertree_online_dlink_nca}),
    each update can be supported in
    $O(\timeCP(D, \sigma))$ time and $O(N)$ space.
  \end{itemize}
\end{corollary}
 \subsection{Suffix-Tree-Based Online Algorithm}
We present an online algorithm for computing distinct palindromes using the suffix tree of a trie.
Only in this subsection, we treat the input trie as a \emph{backward trie},
i.e., we read path labels in the leaf-to-root direction.

The suffix tree $\mathsf{ST}$ of a backward trie $\mathcal{T} = (V, E)$ is
the compacted trie of all suffixes of the trie~\cite{Kosaraju89a}.
When reading path labels from leaves to the root, these suffixes are exactly the strings
$\str(v, r)$ for $v \in V$.
Throughout this subsection, we use bold symbols to denote nodes in $\mathsf{ST}$.
We assume that the root of $\mathcal{T}$ has a single outgoing edge labeled with a special character \$, which does not appear elsewhere in $\mathcal{T}$.
This assumption ensures that $\mathsf{ST}$ has a leaf corresponding to each node in $\mathcal{T}$.
For each node $\mathbf{u}$ in $\mathsf{ST}$, let $\str(\mathbf{u})$ denote the path string spelled out by the root-to-$\mathbf{u}$ path in $\mathsf{ST}$.
For each node $x$ in $\mathcal{T}$, we store a bidirectional pointer to a leaf $\mathbf{x}$ in $\mathsf{ST}$ such that $\str(\mathbf{x}) = \str(x, r)$.

To determine which operations are needed to detect new palindromes when a new leaf is inserted,
we observe from the proof of Lemma~\ref{lem:num_distinct_pal_trie} that only one new palindrome can appear,
and it must be the longest palindromic suffix of the string corresponding to the new leaf.
Thus, the following two operations for each new leaf suffice to achieve the online computation of distinct palindromes in a trie:
\begin{enumerate}
  \item Compute the longest palindromic suffix of the new leaf $v$ in the trie.
  \item Determine whether the palindrome is unique in the current trie and, if so, report it as a distinct palindrome.
\end{enumerate}
Since the longest palindromic suffix of the new leaf $v$ is always maximal,
we can compute such a palindrome in $O(\min(\log h, \sigma))$ time by Corollary~\ref{col:dynamicMPal}.
To determine the uniqueness of the palindrome, we use the suffix tree $\mathsf{ST}$ of the backward trie.
Let $\mathbf{v}$ be the node in $\mathsf{ST}$ corresponding to the new leaf $v$ in the input trie $\mathcal{T}$.
Then the parent of $\mathbf{v}$ in $\mathsf{ST}$ corresponds to the longest repeating suffix of $\str(r, v)$ in $\mathcal{T}$, where $r$ is the root.
Therefore, we can determine the uniqueness of the longest palindromic suffix $P$
by comparing its length $|P|$ with $|\str(\parent(\mathbf{v}))|$.
Specifically, $|\str(\parent(\mathbf{v}))| < |P|$ if and only if $P$ is unique in $\mathcal{T}$.
Since all these operations can be performed in $O(\min(\log h, \sigma))$ time,
maintaining the suffix tree of the backward trie yields an online algorithm for computing distinct palindromes with an additional $O(\min(\log h, \sigma))$ overhead per insertion.

In the following, we will show the following result.

\begin{restatable}{theorem}{SuffixTreeConstruction} \label{thm:strie}
  The suffix tree for a trie $\mathcal{T}$ can be constructed in an online manner in $O(N \timeCPincr(N, \sigma))$ time and $O(N)$ space.
\end{restatable}
Therefore, we obtain the following theorem:
\begin{restatable}{corollary}{ComputingDistinctPals} \label{cor:dist_pal}
  We can compute all distinct palindromes in a trie $\mathcal{T}$ in an online manner in $O(N (\timeCPincr(N, \sigma) + \min(\log h, \sigma)))$ time and $O(N)$ space.
\end{restatable}

\paragraph*{\bf Online Construction of Suffix Trees for a Trie}\label{sec:strie}
Below, we give a proof of Theorem~\ref{thm:strie}.

We assume that, for each node $\mathbf{u}$ in $\mathsf{ST}$ and
each node $\mathbf{v}$ satisfying $\str(\mathbf{v}) = c \cdot \str(\mathbf{u})$ for some $c \in \Sigma$,
the node $\mathbf{u}$ stores a link labeled with $c$ to the node $\mathbf{v}$.
In this case, we say that $\mathbf{u}$ is marked by the character $c$.
A single node can be marked by multiple characters and may have multiple links labeled with distinct characters.
Since each link from $\mathbf{u}$ corresponds to a unique node $\mathbf{v}$ such that
$\str(\mathbf{v})[2..|\str(\mathbf{v})|] = \str(\mathbf{u})$,
the total number of both marks and links is at most $N$.

We focus on an existing right-to-left online construction algorithm for the suffix tree of a single string, proposed in~\cite{KucherovN17_FullFledged}.
The suffix tree of a single string $X$ is equivalent to the suffix tree of the trie consisting of only a single leaf representing $X$.
Hence, this setting is a special case of constructing a suffix tree for a backward trie online.
We first review the algorithm of \cite{KucherovN17_FullFledged} for a single string.
Then, we show that this algorithm can be adapted to construct a suffix tree for a trie.

Consider the suffix tree $\mathsf{ST}$ of a single string $X$, and suppose we prepend a new character $c$ to $X$, forming $Y = cX$.
We then need to insert a new leaf $\mathbf{y}$ corresponding to $Y$ into $\mathsf{ST}$, and if necessary, insert the parent node $\mathbf{y}'$ of $\mathbf{y}$.
Let $\mathbf{x}'$ be the node satisfying $\str(\mathbf{y}') = c \cdot \str(\mathbf{x}')$.
Note that $\mathbf{x}'$ is the parent of the node $\mathbf{x}$ corresponding to $X$,
and such a node $\mathbf{x}'$ always exists as a branching node in $\mathsf{ST}$ prior to the insertion of new nodes (see Fig.~\ref{fig:strie}).
A crucial step of the algorithm is to find the position in the $\mathsf{ST}$ where new nodes should be inserted.
This step is guided by the Euler tour of $\mathsf{ST}$.
Let $\mathbf{z}_1$ and $\mathbf{z}_2$ be the closest nodes preceding and following $\mathbf{x}$ in the Euler tour that are marked by the character $c$.
Then, $\mathbf{x}'$ is the lower of $\lca(\mathbf{x}, \mathbf{z}_1)$ and $\lca(\mathbf{x}, \mathbf{z}_2)$,
where $\lca(\mathbf{a}, \mathbf{b})$ is the lowest common ancestor of $\mathbf{a}$ and $\mathbf{b}$ in $\mathsf{ST}$.
Moreover, if $\lca(\mathbf{x}, \mathbf{z}_1)$ is the lowest, then the node $\mathbf{w}_1$ with $\str(\mathbf{w}) = c \cdot \str(\mathbf{z}_1)$
is either the parent of the new node $\mathbf{y}$ or the endpoint of the edge split by inserting the new node.
If $\lca(\mathbf{x}, \mathbf{z}_2)$ is the lowest, the node $\mathbf{w}_2$ with $\str(\mathbf{w}_2) = c \cdot \str(\mathbf{z}_2)$ has the same role~\cite[Lemma 5]{KucherovN17_FullFledged}.
Thanks to this property, we can compute the insertion point of the new edge by finding $\mathbf{z}_1$, $\mathbf{z}_2$, and their lowest common ancestor with $\mathbf{x}$.

This algorithm requires two main operations:
(1) finding the closest node in the Euler tour marked by a specific character, and
(2) computing the lowest common ancestor of two nodes in $\mathsf{ST}$.
The first operation can be reduced to a colored predecessor problem, described in Section~\ref{sec:colored_pred}.
We store all marks in $\mathsf{ST}$ in $\mathcal{L}$ sorted by the Euler tour order
and use the colored predecessor data structure to find the closest node with a specific character.
Since the total number of marks is at most $n$, the space complexity is $O(|X|)$.
For the second operation, we use a dynamic LCA data structure~\cite{ColeH05_dynamicLCA}
to answer queries in constant time.
This results in an overall update time of $O(\timeCPincr(|X|, \sigma))$.

The proof of the above property relies on the following fact:
for any pair of nodes $\mathbf{u}$ and $\mathbf{v}$ such that
$\str(\mathbf{v}) = c \cdot \str(\mathbf{u})$ for a character $c$, the subtree rooted at $\mathbf{v}$ is isomorphic to the graph obtained from the subtree rooted at $\mathbf{u}$ by contracting all nodes not marked by $c$.
We now justify that the same fact holds for the suffix tree of a backward trie.

\begin{lemma}\label{lem:isomorphic_subtrees_trie}
Let $\mathsf{ST}$ be the suffix tree of a backward trie $\mathcal{T}$.
For any character $c \in \Sigma$ and any pair of nodes $\mathbf{u}, \mathbf{v}$ in $\mathsf{ST}$ satisfying
$\str(\mathbf{v}) = c \cdot \str(\mathbf{u})$,
the subtree rooted at $\mathbf{v}$ is isomorphic to the graph obtained from the subtree rooted at $\mathbf{u}$ by
restricting it to the nodes marked by $c$ and contracting all unmarked nodes.
\end{lemma}
\begin{proof}
For a node $\mathbf{x}$ of $\mathsf{ST}$, let $\Leaves(\mathbf{x})$ denote the set of leaves in the subtree rooted at $\mathbf{x}$.
By construction of $\mathsf{ST}$, a leaf $\mathbf{y}$ belongs to $\Leaves(\mathbf{x})$ if and only if $\str(\mathbf{x})$ is a prefix of $\str(\mathbf{y})$.
Since each leaf $\mathbf{y}$ corresponds to exactly one trie node $y\in V$ with $\str(\mathbf{y})=\str(y,r)$, the set $\Leaves(\mathbf{x})$ corresponds to the set of trie nodes $y$ such that $\str(\mathbf{x})$ is a prefix of $\str(y,r)$.

Now fix $c\in\Sigma$ and nodes $\mathbf{u},\mathbf{v}$ with $\str(\mathbf{v})=c\cdot\str(\mathbf{u})$.
Consider any $\mathbf{x}\in\Leaves(\mathbf{v})$, and denote $\str(\mathbf{x})=\str(\mathbf{v})\cdot \alpha=c\cdot\str(\mathbf{u})\cdot \alpha$ for some string $\alpha$.
Let $x\in V$ be the trie node corresponding to $\mathbf{x}$.
Removing the first character of $\str(x,r)$ corresponds to moving one step toward the root in the backward trie.
Therefore, there exists a trie node $y$ such that $\str(y,r)=\str(\mathbf{u})\cdot\alpha$.
Let $\mathbf{y}$ be the leaf corresponding to $y$ in $\mathsf{ST}$.
Then $\mathbf{y}\in\Leaves(\mathbf{u})$ and $\str(\mathbf{x})=c\cdot\str(\mathbf{y})$, so $\mathbf{y}$ is marked by $c$.
This gives an injective mapping from $\Leaves(\mathbf{v})$ to $\Leaves(\mathbf{u})$.
The image of this mapping is exactly the set of leaves in the subtree rooted at $\mathbf{u}$ that are marked by $c$.
Let $\mathbf{y}'\in\Leaves(\mathbf{u})$ be a leaf marked by $c$.
By the definition of marking, there exists a unique leaf $\mathbf{x}'$ with $\str(\mathbf{x}')=c\cdot\str(\mathbf{y}')$.
Since $\str(\mathbf{u})$ is a prefix of $\str(\mathbf{y}')$, it follows that $\str(\mathbf{v})=c\cdot\str(\mathbf{u})$ is a prefix of $\str(\mathbf{x}')$.
Hence, $\mathbf{x}'\in\Leaves(\mathbf{v})$.

Consequently, the set of strings obtained from the leaves under $\mathbf{v}$ by deleting the common prefix $\str(\mathbf{v})$
is identical to the multiset obtained from the leaves under $\mathbf{u}$ that are marked by $c$ by deleting the common prefix $\str(\mathbf{u})$.
The subtree rooted at a node of the suffix tree is the compacted trie of these strings.
If we restrict the subtree rooted at $\mathbf{u}$ to the nodes marked by $c$ and then contract all unmarked nodes,
the resulting graph is exactly this compacted trie, with paths that do not lead to leaves marked by $c$ removed.
Therefore, the subtree rooted at $\mathbf{v}$ is isomorphic to the graph obtained from the subtree rooted at $\mathbf{u}$
by restricting it to the nodes marked by $c$ and contracting all unmarked nodes.
\end{proof}

By the above lemma,
we can use the same algorithm for the online construction of a suffix tree over a trie.
As mentioned earlier, the total number of marks is at most $N$.
Therefore, the space complexity of this algorithm is $O(N)$.
Using this approach, we achieve an update time of $O(\timeCPincr(N, \sigma))$ for the suffix tree of the trie in $O(N)$ space.
We note that both child traversal and insertion can also be supported within the same time bounds
using appropriate predecessor data structures~\cite{NavarroR20_PredecessorSearch}.
As a result, we obtain the following theorems.

\SuffixTreeConstruction*

\ComputingDistinctPals*

The above data structure can be extended to support leaf deletions.  
When a leaf node $y$ is deleted from the trie, the corresponding leaf $\mathbf{y}$ in the suffix tree $\mathsf{ST}$ must also be removed.  
If the parent $\mathbf{y}'$ of $\mathbf{y}$ becomes non-branching as a result, it should likewise be deleted.  
Since each leaf node in the trie maintains a pointer to its corresponding leaf in $\mathsf{ST}$, the location of $\mathbf{y}$ can be identified in constant time.  
Moreover, determining whether $\mathbf{y}'$ should be removed is straightforward.  
Therefore, by using a colored predecessor data structure that supports $\Delete$ and $\Uncolor$,
we can implement both insertions and deletions.
We summarize the result in the following.

\begin{theorem}\label{thm:strie_del}
  The suffix tree for a trie $\mathcal{T}$
  can be maintained online with support for leaf insertions and leaf deletions
  in $O(\timeCP(N,\sigma))$ time per update using $O(N)$ space.
\end{theorem}

\begin{corollary}\label{cor:dist_pal_del}
  The set of distinct palindromes in a trie $\mathcal{T}$,
  which supports both leaf insertions and leaf deletions,
  can be maintained online in 
  $O(\timeCP(N,\sigma) + \min(\log h, \sigma))$ time per update
  using $O(N)$ space.
\end{corollary}

\begin{figure}[h]
\begin{minipage}{0.5\textwidth}
    \centering
    \includegraphics[scale=0.4]{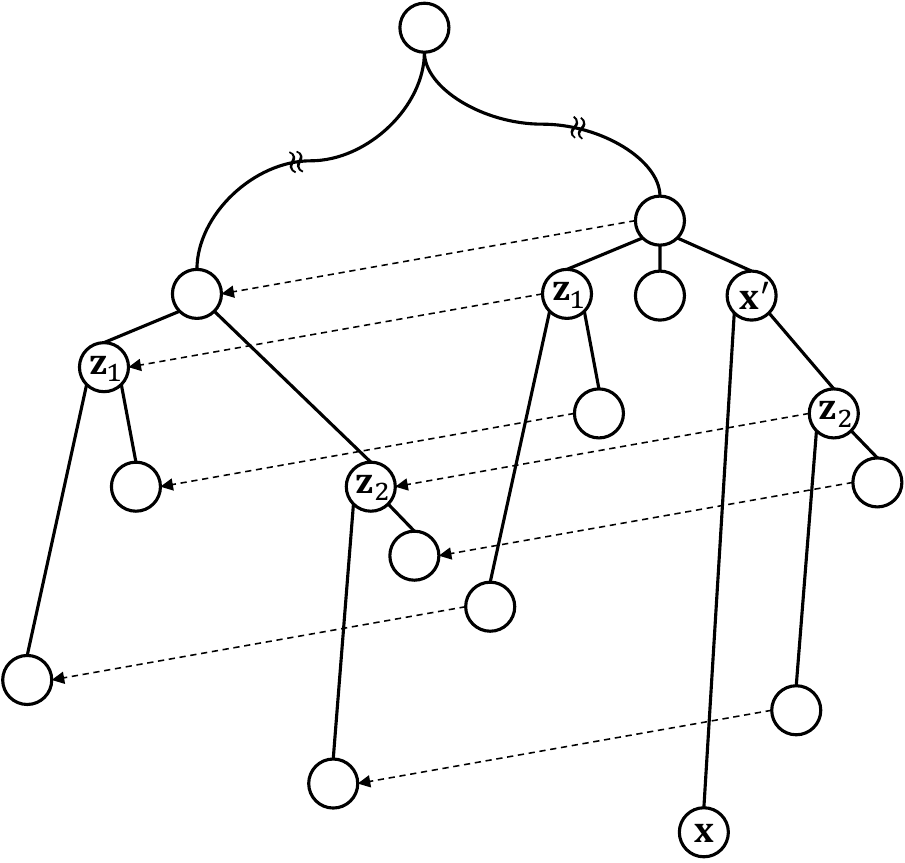}
  \end{minipage}
  \begin{minipage}{0.5\textwidth}
    \centering
    \includegraphics[scale=0.4]{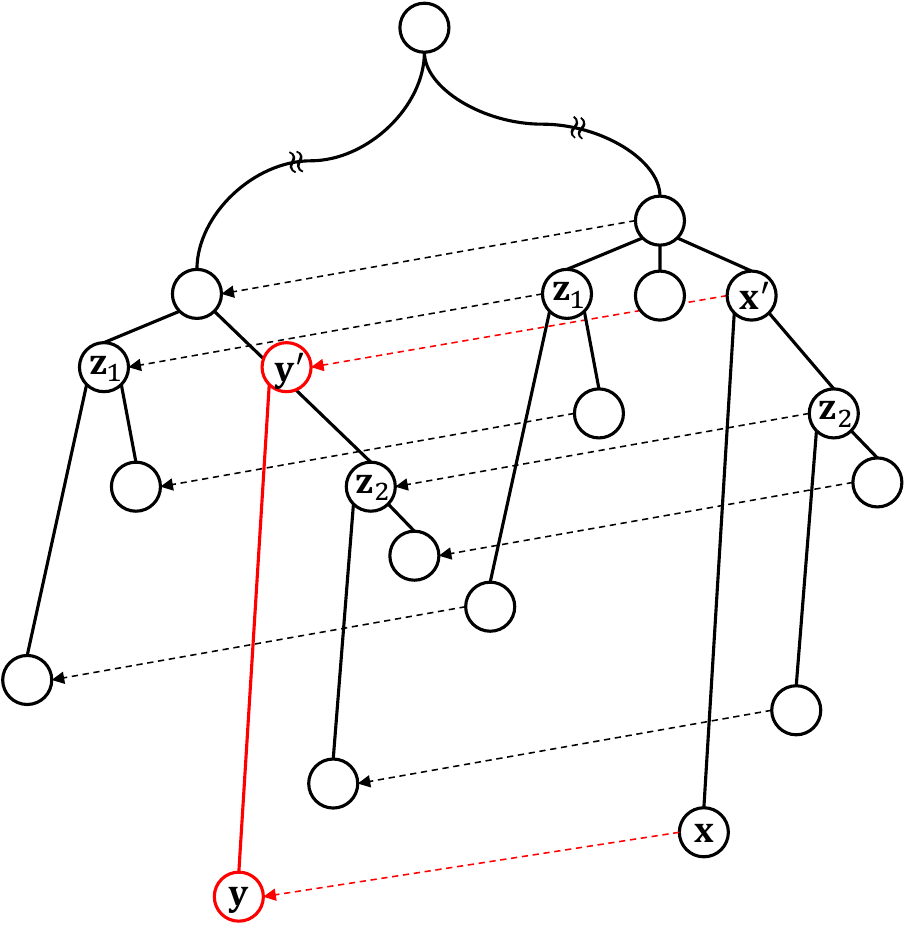}
  \end{minipage}
\caption{
    An illustration of the suffix tree $\mathsf{ST}$ of a single string before (left) and after (right)
    prepending a new character $c$ to $X$, forming $Y = cX$.
    Nodes and the edge marked in red have to be added in $\mathsf{ST}$.
    A dashed arrow represents the link labeled with $c$.
    Note that each starting node of a dashed arrow is marked by $c$.
    In this figure, $\mathbf{x}' = \lca(\mathbf{x}, \mathbf{z}_2)$ holds, and $\mathbf{w}_2$ is the endpoint of the edge split by inserting the new node $\mathbf{y}'$.
  }
  \label{fig:strie}
\end{figure}
 \section{Conclusions and Future Work}\label{sec:conclusions}

In this paper,
we studied the problem of computing palindromes in a trie
that supports leaf insertions and leaf deletions.
We considered both maximal palindromes and distinct palindromes.

For maximal palindromes,
we presented an algorithm that runs in $O(N \log h)$ time
and works over a general unordered alphabet.
The algorithm maintains all maximal palindromes efficiently
under leaf insertions and deletions in the trie.

For distinct palindromes,
we proposed several online algorithms based on two different frameworks:
one using the suffix tree of a backward trie,
and the other using the EERTREE of a forward trie.
These algorithms support leaf insertions and deletions
and maintain the set of distinct palindromes under these updates.
As a by-product,
we showed an efficient online algorithm for constructing
the suffix tree and the EERTREE of a trie, which is of independent interest.

As future work, the following problems will be interesting:
\begin{itemize}
  \item
    Can we extend our results on semi-dynamic trees to support root additions and edge-label substitutions?
    A na\"ive extension of Manacher's algorithm on the trie with $L$ leaves spends $O(L)$ time for each root addition
    since there are at most $L$ candidates for the palindromic suffixes in the backward trie
    and they cannot be represented compactly in the same way as for the palindromic prefixes.
  \item
    Can we efficiently compute palindromes in an edge-labeled DAG?
    An edge-labeled DAG is a natural generalization of a trie.
    Computing palindromes in a general DAG seems challenging.
    Perhaps a first step is to investigate properties of palindromes in
    some reasonable subclasses of edge-labeled DAGs.
  \item
    Can we construct the suffix tree of a trie in $O(N \log \sigma)$ time online?
    Our current approach relies on the colored predecessor data structure, which seems to limit further improvements in running time.
    Since our method is based on real-time construction of the suffix tree and no faster real-time online construction algorithm is known,
    achieving such a bound may require a fundamentally different approach.
\end{itemize}
 
\section*{Acknowledgements}
The authors thank the anonymous reviewers of earlier versions of this article
for their helpful comments.
This work was supported by JSPS Grant Numbers
JP24K20734 (TM), JP23K24808 (SI), JP24K02899 (HB).
\bibliographystyle{unsrt}
\bibliography{ref}
\end{document}